\documentclass[pra,reprint,showpacs,superscriptaddress]{revtex4-1}
\usepackage[T1]{fontenc}
\usepackage{amsmath, amsthm, amssymb, mathtools, dsfont}
\usepackage{times}
\usepackage{graphicx}
\usepackage{dcolumn}
\usepackage{bm,bbm}
\usepackage[normalem]{ulem}
\usepackage{color}
\usepackage[usenames,dvipsnames]{xcolor}
\usepackage{esint}
\usepackage{paralist}

\usepackage{geometry}
\definecolor{myurlcolor}{rgb}{0,0,0.7}
\definecolor{myrefcolor}{rgb}{0.8,0,0}
\usepackage[unicode=true,pdfusetitle, bookmarks=true,bookmarksnumbered=false,bookmarksopen=false, breaklinks=false,pdfborder={0 0 0},backref=false,colorlinks=true, linkcolor=myrefcolor,citecolor=myurlcolor,urlcolor=myurlcolor]{hyperref}

\DeclareGraphicsExtensions{.png,.pdf,.eps}
\graphicspath{ {./figures/} }
\usepackage{epstopdf}

\makeatletter

 \theoremstyle{plain}
 
 \theoremstyle{plain}
 \newtheorem{lem}{Lemma}
 \theoremstyle{plain}
 \newtheorem{thm}{Theorem}
 \theoremstyle{plain}
 \newtheorem{exa}{Example}
 \theoremstyle{plain}
 
 \theoremstyle{plain}
 
  \theoremstyle{plain}
 \newtheorem{defn}{Definition}
 \theoremstyle{remark}
 \newtheorem*{rem*}{Remark}
 \theoremstyle{plain}
  
\newcommand{\scalar}[2]{\langle #1 | #2 \rangle}
\newcommand{\ketbra}[2]{| #1 \rangle \langle #2 |}
\newcommand{\ket}[1]{| #1 \rangle}
\newcommand{\bra}[1]{\langle #1 |}

\newcommand{\1}{\mathbbm{1}} 

\renewcommand{\exp}{\mathrm{exp}}
\DeclareMathOperator{\tr}{tr}

\renewcommand{\H}{\mathcal{H}}

\newcommand{\defeq}{\coloneqq}

\renewcommand{\C}{\mathbb{C}} 
\newcommand{\M}{\mathbf{M}} 
\newcommand{\N}{\mathbf{N}} 
\renewcommand{\Pr}{\mathrm{Pr}} 

\renewcommand{\P}{\mathbf{P}} 
\newcommand{\PP}{\mathcal{P}} 

\newcommand{\rbracket}[1]{\left(#1\right)} 
\newcommand{\cbracket}[1]{\left\{#1\right\}} 

\newcommand{\eq}[1]{Eq.~\eqref{#1}}

\newcommand{\SP}{\mathbb{S}\mathbb{P}} 

\global\long\global\long\global\long\def\bk#1#2{\mbox{\ensuremath{\ensuremath{\langle#1|#2\rangle}}}}
\global\long\global\long\global\long\def\kb#1#2{\mbox{\ensuremath{\ensuremath{\ensuremath{|#1\rangle\!\langle#2|}}}}}

\begin{document}

\title{Simulating all quantum measurements using only projective measurements and postselection}

\author{Micha\l\ Oszmaniec}
\email{michal.oszmaniec@gmail.com}
\affiliation{ 
Institute of Theoretical Physics and Astrophysics, National Quantum Information Centre, Faculty of Mathematics, Physics and
Informatics, University of Gdansk, Wita Stwosza 57, 80-308 Gda\'nsk, Poland}

\author{Filip B. Maciejewski}
\email{filip.b.maciejewski@gmail.com}
\affiliation{University of Warsaw, Faculty of Physics, Ludwika Pasteura 5, 02-093 Warszawa, Poland}

\author{Zbigniew Pucha{\l}a}
\email{z.puchala@iitis.pl}
\affiliation{Institute of Theoretical and Applied Informatics, Polish Academy
of Sciences, ulica Ba{\l}tycka 5, 44-100 Gliwice, Poland}
\affiliation{Faculty of Physics, Astronomy and Applied Computer Science, 
Jagiellonian University, ul. {\L}ojasiewicza 11,  30-348 Krak{\'o}w, Poland}

\begin{abstract}
We report an alternative scheme for implementing generalized quantum measurements that does not require the usage of auxiliary system. Our method utilizes solely: (a) classical randomness and post-processing, (b) projective measurements on a relevant quantum system and (c) postselection on non-observing certain outcomes. The scheme implements arbitrary quantum measurement in dimension $d$ with the optimal success probability $1/d$. We apply our results to bound the relative power of projective and generalised measurements for unambiguous state discrimination. Finally, we test our scheme experimentally on IBM's quantum processor. Interestingly, due to noise involved in the implementation of entangling gates, the quality with which our scheme implements generalized qubit measurements outperforms the standard construction using the auxiliary system.
\end{abstract}
\maketitle

\section{Introduction}
Every quantum information or quantum computing protocol contains,
as a subroutine, a measurement of the quantum state. 
Quantum theory admits
measurement procedures that are more general then the commonly known
projective measurements (PMs) of a quantum observable
on a quantum system of interest. Indeed, the most general quantum measurements
can be realized by projective measurements on a system extended by
the suitable ancilla \cite{Peres2006}. Such generalized measurements
are mathematically described by Positive Operator-Valued Measures (POVMs)
and play important role in many areas of quantum information science
such as quantum tomography \cite{Renes2003,OptTomography2009}, state discrimination \cite{Chefles2000,Barnett09,StructureStateDISC},
(multi-parameter) quantum metrology \cite{MultiparamMetroRev} or
quantum computing \cite{opotHSP}. They are also relevant in studies
of foundations on quantum theory \cite{Busch2003,ContextualitySpekkens},
nonlocality \cite{Barrett2002,Brunner2014,Kleinmann2016a,Kleinmann2016b} and randomness generation
\cite{Acin2016}.

Projective measurements form a subset of POVMs and hence are generally
less powerful for information processing. However,  there are two issues that need to be addressed before one decides to implement generalized measurements in practice.  The first problem is that POVMs
are often difficult to realize as their implementation 
typically requires control and manipulation over additional degrees
of freedom \cite{Exp2009,FoundChile2011,Ota2012} (such as, for example,
path in the case of quantum states encoded in photon polarisation). The second problem is that the relative power of projective and generalized measurements for quantum information processing remains poorly understood \footnote{There exist problems for which PMs
offer the optimal solution. The examples include: single parameter
metrology \cite{Giovannetti2011}, discrimination between two quantum
states \cite{Nielsen2010}, minimal error state discrimination of
linearly independent pure states \cite{EldarMESD2003}, or certain
variants of random access codes \cite{Ambainis2008}.},
especially for Hilbert spaces of large dimension (see however
 \cite{Oszmaniec2017,Hirsch2017,Leo2017,Oszmaniec2019,Uola2018,Takagi2019}). The main aim of this work is to provide an alternative method for implementation of generalized measurements and to advance the understanding of the relative power of POVMs and PMs.

We start by presenting a new scheme that realizes arbitrary
POVM without the need to extend the Hilbert
space \footnote{In some physical systems certain  POVMs are easier to implement than PMs. In particular, in quantum optics \cite{bachor2018guide}   homodyne and heterodyne measurements effectively implement nontrivial POVMs on the Fock space describing quantum states of light.}. Specifically, our method uses only (a) classical randomness
and post-processing, (b) PMs (acting only on
a Hilbert space of interest) and (c) postselection on non observing
certain measurement outcomes. The price that we need to pay is that
in a given experimental run, the measurement is carried out with
success probability $1/d$. We prove that this number is optimal in
a sense that there always exist measurements for
which success probability cannot be higher. Our method can be regarded as the manifestation of the following trade-off. Namely, in order
to implement a generalized measurement an experimentalist can
either implement a complicated PM on a system coupled to the
ancilla or implement simpler PMs and apply postselection.

Moreover, we use our method to give insight into
the question of relative power between projective and generalized
measurements for unambiguous state discrimination (USD) \cite{DIEKS1988}. Specifically, we show
that in this scenario the ratio between optimal discrimination probabilities,
when using POVMs and projective measurements, is at most $d$. We also give examples of ensembles of states for
which this bound is essentially optimal.

Finally, we demonstrate our method experimentally on IBM's quantum processor
\cite{ibm_q_experience,qiskit,cross_qasm}. 
We implement generalized qubit POVMs via
our scheme and via the Naimark construction \cite{Peres2006} that uses PMs 
on two qubits. We compare the quality of two implementations by performing
tomography of measurement operators. Interestingly, due to noise involved in
implementation of entangling gates, the quality with which our scheme
realizes POVMs is higher than the one obtained with the Naimark method.

The the rest of paper is structured as follows. 
First, in Section \ref{sec:Not} we establish the notation and main concepts that will be used in the rest of the paper. 
Second, in Section \ref{sec:postSIM} we present a notion of simulation of quantum measurements via postselection and projective measurements. 
This part contains also the main results of our work (Theorems \ref{thm:mainRES}  and \ref{thm:mainRES2}). 
In the following Section \ref{sec:USD} we give application of our simulation scheme to the problem of relative advantage of projective and generalized quantum measurements for the problem of unambiguous state discrimination. 
In Section \ref{sec:IBM} we provide experimental illustration of our findings on IBM's quantum device.
The Section \ref{sec:Concl} contains concluding remarks and further  research directions.  
We complement the main part of the article with Appendix \ref{app::details_ex_3} (containing proofs of some technical statements omitted in the main text), and  Appendices \ref{app:QMT}, \ref{app:povms_explicit}  and \ref{app:circuits} (containing details of the experimental procedures on IBM Q experience).

\section{Notation nad main concepts}\label{sec:Not}

 A $n$-outcome POVM on $d$-dimensional space can be regarded as a vector $\M=\left(M_{1},\ldots,M_{n}\right)$ of non-negative operators satisfying $\sum_{i=1}^{n}M_{i}=\1$, where $\1$
is the identity on $\C^d$. The operators $(\M)_i \defeq M_{i}$ are called the effects of $\M$. 
According to Born's rule, when the quantum state $\rho$ is measured by a POVM $\M$, the probability of obtaining the outcome $i$ is given by $\Pr(i|\rho,\M)=\mathrm{tr}\left(M_{i}\rho\right)$. 
We denote the set of POVMs on $\C^d$ with $n$ outcomes by $\PP\left(d,n\right)$. 
Given two POVMs $\M,\N\in\PP\left(d,n\right)$, their convex combination $p\M+(1-p)\N$ is the POVM with $i$-th effect given by $\left[p\M+\left(1-p\right)\N\right]_i \defeq pM_{i}+(1-p)N_{i}$.
 Taking convex combinations of measurements is typically referred to as \emph{randomisation} as it corresponds to realizing POVMs $\M$ and $\N$ with certain probabilities and then combining the outcomes. Extremal POVMs are the measurements that cannot be expressed by a convex combination of two different POVMs. PMs are POVMs whose effects are orthogonal projectors.
(notice that some of the outputs can have null effects and that effects are not required to be rank-one).

In \cite{Oszmaniec2017} the class of \emph{projective simulable} measurements was introduced. 
By definition, measurements belonging to this class can be realized by randomisation followed by classical post-processing (see \cite{Buscemi2005,Haapasalo2012,Sentis2013} for a detailed exposition of these concepts) of some protective measurements $\P$ acting on $\C^d$ alone. Short mathematical description of those notions is as follows. Let's consider the POVM $\M$ which has a convex decomposition $\M=\sum_i p_i \P^i$ into some projective measurements $\cbracket{\P^i}$, with $\sum_i p_i=1$ and $p_i\geq0$ for all $i$.
Set $\cbracket{p_i}$ may be interpreted as a probability distribution for randomization of $\cbracket{\P^i}$.
Randomized implementation of $\cbracket{\P^i}$ leads to the simulation of $\M$ by projective measurements $\cbracket{\P^i}$, so far without any post-processing.
Consider now a measurement $\N$, which is related to $\M$ via post-processing, i.e., $N_j=\sum_k q(j|k)M_k $, where the numbers $\cbracket{q\rbracket{j|k}}$ correspond to particular post-processing strategy, with $\sum_j q\rbracket{j|k}=1$ and $q\rbracket{j|k}\geq0$ for all $j,k$.
Therefore, randomization of $\cbracket{\P^i}$ according to probability distribution $\cbracket{p_i}$, followed by post-processing specified by $\cbracket{q\rbracket{j|k}}$, realizes the simulation of $\N$ by projective measurements $\cbracket{\P^i}$. Such simulation should be understood in terms of sampling from the statistics that one would have obtained if the POVM $\N$ (or $\M$) was implemented directly.
We denote the class of projective simulable $n$-outcome POVMs on $\C^d$ by $\SP(d,n)$  . 
Clearly, no ancillary system or extra dimension are needed to realize projective simulable measurements.
However, not all measurements can be implemented in this manner. 
In particular, all extremal but not projective measurements are outside $\SP(d,n)$ \cite{Oszmaniec2017}. 

\section{Simulation of measurements with postselection}\label{sec:postSIM}

We will be interested in measurements that can be realized by projective simulable measurements together with postselection.

\begin{defn}[Simulation of POVMs by postselection]
 \label{def:postSEL}Let $\M\in\PP(d,n)$ and $\N\in\PP(d,n')$
be $n$- and $n'$-outcome POVMs on $\mathbb{C}^{d}$ and let $n'>n$.
We say that $\M$ can be simulated by $\N$ by postselection
if for all quantum states $\rho$ and for all $i\leq n$
\begin{equation}
\Pr(i|\rho,\M)=\frac{\Pr(i|\rho,\N)}{\Pr(i\leq n|\rho,\N)}\,.\label{eq:defPOST}
\end{equation}
In other words, all the statistics of measurement of $\M$
can be interpreted as statistics of $\N$ \textit{conditioned}
on not observing particular outcomes. 
\end{defn}

While Definition \ref{def:postSEL} is operationally well-motivated, it is also cumbersome to work with.
Let us introduce a measurement $\M_{q}\in\PP(d,n+1)$ via 
\begin{equation}
\M_{q}\coloneqq\left(qM_{1},\ldots,qM_{n},\left(1-q\right)\1\right)\ .\label{eq:defMQ}.
\end{equation}
 It turns out that Eq.(\ref{eq:defPOST}) is equivalent to the existence of $q\in(0,1]$ such that
\begin{equation}
\M_{q}=\left(N_{1},\ldots,N_{n},\sum_{i>n}N_{i}\right)\ .\label{eq:equivPOST}.
\end{equation}
\begin{proof}
	Clearly, Eq.(\ref{eq:equivPOST}) implies Eq.(\ref{eq:defPOST}). 
	We will now prove the reverse.
	Let $\M\in\PP(d,n) $ and $\N\in\PP(d,n')$ be POVMs on $\C^d$ such that for all quantum states $\rho$ and for all $i\leq n$, the \eq{eq:defPOST}
	We claim now that if \eqref{eq:defPOST} holds then  $\Pr(i\leq n| \rho, \N)$ does not depend on $\rho$.  
	To see this we consider a state $\rho = \alpha \sigma + (1-\alpha) \tau$, for some $0\leq \alpha \leq 1$. By the linearity of Born rule we obtain
	
	\begin{equation}
	\begin{split}
	\Pr(i| \rho, \M) 
	= 
	\frac{
		\alpha \Pr(i| \sigma, \N) + (1-\alpha) \Pr(i| \tau, 
		\N)
	}{
		\alpha \Pr(i\leq n| \sigma, \N) + (1-\alpha) \Pr(i 
		\leq n| \tau, 
		\N)}
	\end{split}\ .
	\end{equation}
	The above, by definition is a linear function of $\alpha$, and it is possible if and only if  $ \Pr(i\leq n| \sigma, \N) =  \Pr(i 
	\leq n| \tau, \N)$. Since states $\sigma$ and $\tau$ were chosen 
	arbitrarily we conclude that $Pr(i\leq n| \rho, \N)$ does not depend 
	on $\rho$. 
It follows that $\sum_{i\leq n}N_{i}=q\1$,
	where $q>0$ is a proportionality constant that can be interpreted
	as the \textit{success probability} of implementing the measurement
	$\M$ via the POVM $\N$.
	
\end{proof}

\begin{figure}
\begin{centering}
\includegraphics[scale=0.5]{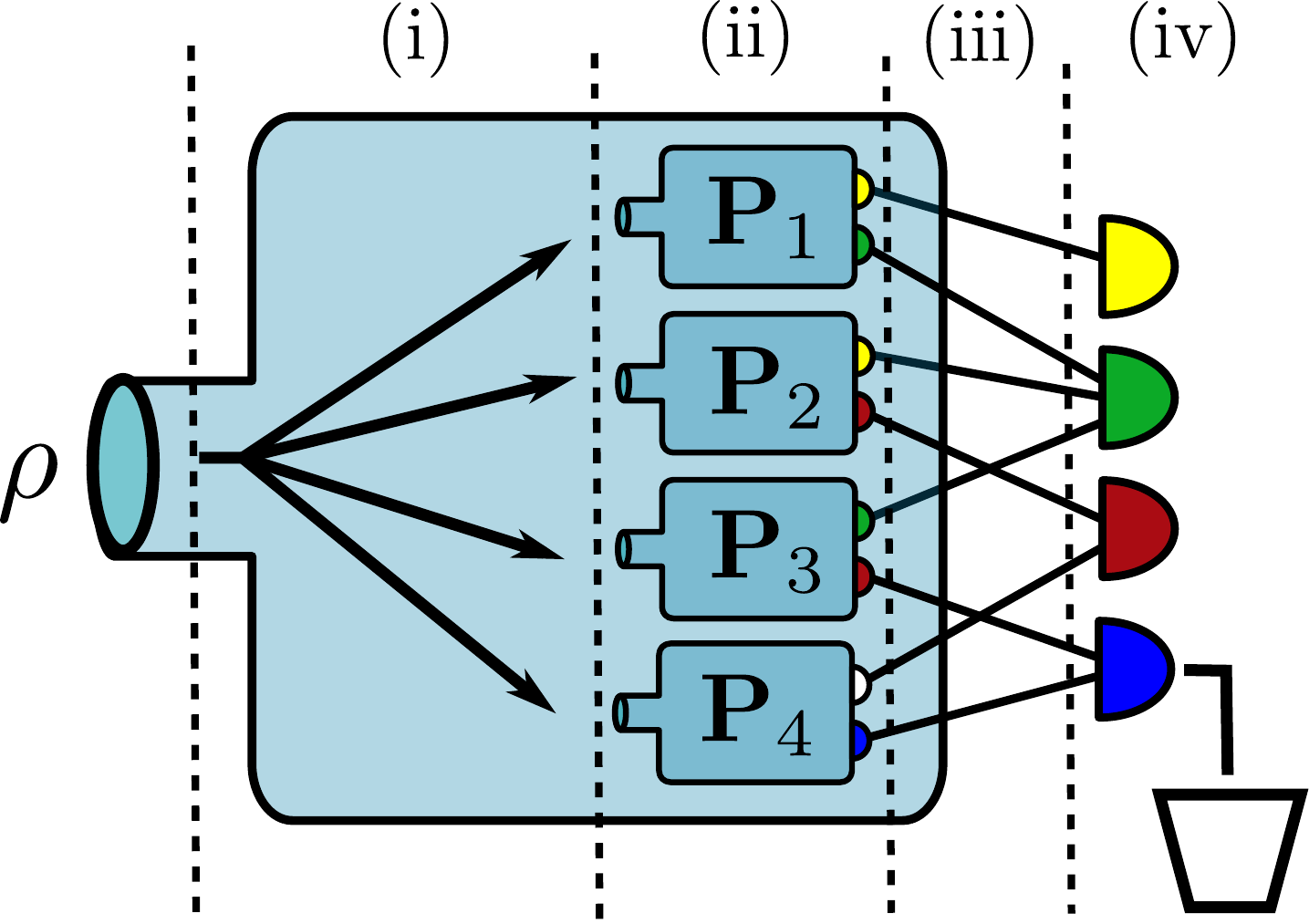}\caption{\label{fig:Schem}The idea of simulation
of a generalized measurement with projective measurements
and postselection. The protocol consists of four stages: (i) classical
randomisation, (ii) projective measurements performed on a relevant
quantum system, (iii) post-processing of obtained outcomes, and finally
(iv) postselection on non-observing the last outcome.}
\par\end{centering}
\end{figure}

We can now give a formal definition of measurement simulable by projective
measurements and postselection \footnote{The set of projective simulable measurements is closed under
post-processing. Therefore, simulation via $\N\in\mathbb{SP}(d,n+k)$
already implies simulation via $\N'\in\mathbb{SP}(d,n+1)$
and hence without loss of generality we limitted ourselves to measurements
with $n+1$ outcomes in Definition \ref{def:mathDEF}.} (see Fig. \ref{fig:Schem}).
\begin{defn}[Quantum measurements simulable by projective measurements and postselection]
\label{def:mathDEF}We say that a POVM $\M\in\PP(d,n)$
can be simulated by projective measurements and postselection if there
exists a projective simulable measurement $\N\in\mathbb{SP}(d,n+1)$
such that Eq.(\ref{eq:defPOST}) holds. Or, to put it differently,
there exists $q>0$ such that $\M_{q}=\N\in\mathbb{SP}(d,n+1).$
The highest number $q$ such that the above condition holds is the\textit{
maximal success probability} with which $\M$ can be implemented
when we are allowed to use only PMs and postselection.
\end{defn}

Recently it was shown \cite{Oszmaniec2017}  that every quantum measurement on finite-dimensional system can be implemented by PM on this system extended by the ancilla of the same dimension. Our first result shows that, somewhat surprisingly, any generalized
measurement in finite-dimensional quantum system can be implemented via PMs and postselection.
\begin{thm}[Every quantum measurement can be simulated by projective measurements
and postselection]
\label{thm:mainRES}Let $\M\in\PP(n,d)$ be a quantum
measurement on $\mathbb{C}^{d}$. Then $\M$ can be simulated
by projective measurements and postselection with success probability
$q=1/d$. In other words we have $\M_{1/d}\in\mathbb{SP}(d,n+1)$.
\end{thm}
\begin{proof}
We prove this result by giving a concrete algorithm
that simulates any generalized measurement with success probability
$1/d$. Note that  it suffices
to give the simulation method for POVM having rank-one effects. Indeed,
any quantum measurement can be obtained from them via classical post-processing
\cite{Barrett2002,Buscemi2005,Haapasalo2012,Sentis2013} (note that not all extremal rank-one measurements are projective). Thus, if $\mathcal{M}_{1/d}\in\mathbb{SP}(d,n+1)$
holds for rank-one measurements, it will also hold for arbitrary measurements
since, by definition, classical post-processing of projective simulable
measurement is still projective simulable. 

The effects of rank-one measurement are of the form $M_{i}=\alpha_{i}\kb{\psi_{i}}{\psi_{i}}$,
with $\alpha_{i}\geq0$. From the condition $\sum_{i=1}^{n}M_{i}=\1$
we get $\sum_{i=1}^{n}\alpha_{i}=d$ and hence numbers $\alpha_{i}/d$
define a probability distribution. The method for simulating a rank
one measurement consists of three steps: (a) draw a label $i$ with
probability $p_{i}=\frac{\alpha_{i}}{d}$; (b) perform a PM $\left(P_{+}^{i},P_{-}^{i}\right)=\left(\kb{\psi_{i}}{\psi_{i}},\1-\kb{\psi_{i}}{\psi_{i}}\right)$;
(c) upon obtaining the outcome ``+'', return $i$, otherwise return
$n+1$. Clearly, this scheme realizes a measurement from $\mathbb{SP}(d,n+1)$.
Moreover, explicit computation shows that it implements $\text{\ensuremath{\M_{1/d}}}$.
\end{proof}
\begin{rem*}
A related protocol appeared in \cite{Hirsch2016} in the
context of deriving local POVM models for certain entangled states.
Also, a similar method was used in \cite{Oszmaniec2017} to simulate
a noisy version of a POVM $\M$, with effects $M_{i}^{'}=(1/d)M_{i}+(1-1/d)\,\frac{\mathrm{tr}(M_{i})}{d}\1$.
The difference between these approaches and the protocol given above
is the last step (iii), in which one identifies the ``wrong'' outcomes
- this allows to simulate any POVM $\M$ with PMs \textit{exactly} once we allow for postselection.
\end{rem*}

\begin{rem*}
Some experimental works \cite{usd_ref_Agnew} simulate statistics of POVM $\M$ by statistics of a number of PMs whose effects are \emph{proportional} to effects of $\M$. This is done by \emph{estimating}  the expectation values  $\tr(\rho \ketbra{\psi_i}{\psi_i} )$ (for a given quantum state $\rho$) and then classically processing the obtained experimental data (in analogy to the procedures performed in standard quantum tomography). We  stress that our method is conceptually different. 
Namely, in \emph{a single experimental run} our scheme either \emph{samples} from the correct probability distribution  or reports failure (see Example \ref{ex::tetrahedral_implementation} for illustration). 
Hence, our method provides a new operational interpretation of generalized quantum measurements.
\end{rem*}

Let us illustrate  the concepts presented above by considering the following example, which provides a detailed account of the simulation of a particular qubit POVM by projective measurements with postselection.
\begin{exa}\label{ex::tetrahedral_implementation}
	Consider a four-outcome tetrahedral measurement on $\mathbb{C}^{2}$, 
	$\M^{\mathrm{tetra}}=\left(\frac{1}{2}\kb{\psi_{1}}{\psi_{1}},\ldots,\frac{1}{2}\kb{\psi_{4}}{\psi_{4}}\right)$,
	where Bloch vectors associated to pure states $\kb{\psi_{i}}{\psi_{i}}$
	form a tetrahedron inscribed in the Bloch sphere \cite{Renes2003}.
	Having $\M^{\mathrm{tetra}}$ written in this form, it is straightforward to construct a new five-outcome measurement $\M^{\mathrm{tetra}}_{\frac{1}{2}}=\left(\frac{1}{4}\kb{\psi_{1}}{\psi_{1}},\ldots,\frac{1}{4}\kb{\psi_{4}}{\psi_{4}},\frac{1}{2}\1\right)$. 
	Note that first four effects are those of $\M^{\mathrm{tetra}}$ multiplied by $\frac{1}{d}$ ($\frac{1}{2}$ for qubits) and the last one is $\rbracket{1-\frac{1}{d}}\1$.  
	One can clearly see that $\M^{\mathrm{tetra}}_{\frac{1}{2}}\in\SP\rbracket{2,5}$.
	Concretely, to simulate $\M^{\mathrm{tetra}}_{\frac{1}{2}}$ by projective measurements, one constructs four PMs of the form $\rbracket{P_{+}^{i},P_{-}^{i}}=\rbracket{\ketbra{\psi_i}{\psi_i},\1-\ketbra{\psi_i}{\psi_i}}$.
	Then in every experimental run, one \emph{draws} a label "$i$" with probability $\frac{1}{4}$ (i.e., according to probability distribution given by coefficients standing next to $\ketbra{\psi_i}{\psi_i}$ in $\M^{\mathrm{tetra}}_{\frac{1}{2}}$, which in this case is uniform) and implements a projective measurement $\rbracket{P_{+}^{i},P_{-}^{i}}$.
	If the obtained outcome is "$+$", one identifies it as "$i$", and if it is "$-$", one identifies it as "$5$".  
	Finally, upon post selecting on not obtaining the $5$th outcome (this happens with probability $1/2$ -- more generally with probability $1/d$ in dimension $d$) we obtain a sample from the measurement $\M^{\mathrm{tetra}}$. 
Therefore, $\M^{\mathrm{tetra}}$  can be simulated by projective measurements with postselection with success probability $1/2$.

\end{exa}

What is the highest success probability for which all measurements on $\mathbb{C}^{d}$ can be implemented with projective
measurements and postselection? Interestingly, for any dimension $d$,
there always exist generalized measurements that cannot be simulated
with probability higher than $1/d$. 
\begin{thm}[Optimality of the simulation protocol]\label{thm:mainRES2}
 For any dimension $d$, there exists a measurement $\M^{\ast}\in\PP(d,d^{2})$
that cannot be simulated by PMs and postselection
with success probability higher than $1/d$ (i.e. $1/d$ is the maximal
$q$ for which $\M_{q}^{\ast}\in\mathbb{SP}(d,d^{2}+1$)).
The protocol presented above attains the success probability $1/d$
and in this sense can be considered optimal. 
\end{thm}
\begin{proof}
In \cite{Adriano2004} it was shown that there exists an extremal
quantum measurement $\M^{\ast}\in\PP(d,d^{2})$ with
$d^{2}$ pairwise non-commuting effects $M_{i}^{\ast}=(1/d)\kb{\psi_{i}}{\psi_{i}}$,
where $\kb{\psi_{i}}{\psi_{i}}$ are suitably-chosen pure states 
\footnote{{Concretely, the states $\ket{\psi_i}$ belong to the orbit of the discrete
Heisenberg-Weyl (generated by $d$- dimensional analogues of $X$
and $Z$ operators) group through the fiducial state $\ket{\psi_0}$, where $\mathcal{C}$
is a normalization constant (see Section 4.1 of \cite{Adriano2004}
for details).}}.
In the rest of the proof, to keep the notation compact, we will
identify $n\equiv d^{2}$. Consider now a modified measurement $\M_{q}^{\ast}$
(see Eq.(\ref{eq:defMQ}) for the definition of $\M_{q}$)
and assume $\M_{q}^{\ast}\in\mathbb{SP}(d,n+1)$ i.e. $\M_{q}^{\ast}$
is a convex combination of PMs $\left\{ \P^{\alpha}\right\} $,
$\M_{q}^{\ast}=\sum_{\alpha}p_{\alpha}\P^{\alpha}$
(recall that in general effects $P^{\alpha}_i$ of $\P^{\alpha}$ need not to be necessarily rank-one). 
Since operators $M_{i}^{\ast}$ are rank-one and effects
of measurements $\P^{\alpha}$ are orthogonal projectors,
for $i\leq n$ we have $P_{i}^{\alpha}=\lambda_{\alpha}\kb{\psi_{i}}{\psi_{i}}$
, with $\lambda_{\alpha}\in\left\{ 0,1\right\} $. In other words,
if for a given $\alpha$ and $i\leq n$ we have $P_{i}^{\alpha}\neq0$,
then necessarily $P_{i}^{\alpha}=\kb{\psi_{i}}{\psi_{i}}$.
As the operators $\kb{\psi_{i}}{\psi_{i}}$ do not commute with each
other, for each $\alpha$ we must have either $\P^{\alpha}=\P^{j}$,
where
\begin{equation}
\P^{j}\defeq\left(\overbrace{0,\ldots,0}^{j-1},\kb{\psi_{j}}{\psi_{j}},0,\ldots,0,\1-\kb{\psi_{j}}{\psi_{j}}\right)\,\label{eq:projDEC}
\end{equation}
for some $j\leq n$ or $\P^{\alpha}=\P^{n}\defeq(0,\ldots0,\1)$.
There are therefore only $n+1$ different PMs
that together  can simulate $\M_{q}^{\ast}$, i.e.,  $\M_{q}^{\ast}=\sum_{j=1}^{n+1}p_{j}\P^{j}$,
for some probability distribution $\left\{ p_{j}\right\} _{j=1}^{n+1}$.
Then, by using $\left(\M_{q}^{\ast}\right)_{j}=p_{j}P_{j}^{j}$
we obtain $q/d=p_{j}$ for $j\leq n$. Finally, from the inequality
$\sum_{j=1}^{n}p_{j}\leq1$ we obtain that $q\leq1/d$.
\end{proof}
\begin{rem*}
Analogous arguments show that all rank-one measurements $\M=\left(a_{1}\kb{\psi_{i}}{\psi_{i}},\ldots,a_{n}\kb{\psi_{n}}{\psi_{n}}\right)$
for which $\bk{\psi_{i}}{\psi_{j}}\neq0$ can be simulated with
PMs and postselection with success probability
at most $1/d$. Therefore, the protocol given in the proof of Theorem
\ref{thm:mainRES} is also optimal for this broad class of measurements. Of course, some measurements can be implemented with higher probability.
\end{rem*}

\section{Application to USD} \label{sec:USD}

We use our findings to limit the maximal advantage that POVMs offer
over projective measurements for unambiguous discrimination of quantum
states \cite{DIEKS1988,Barnett09,StructureStateDISC}. This task is about unambiguously
discriminating between (not necessarily orthogonal) signal states $\left\{ \rho_{i}\right\} _{i=1}^{n}$, each appearing with probability $p_{i}$. 
The problem of USD is the landmark example of the task for which POVMs offer advantage over projective measurements. It currently finds applications in quantum cryptography  \cite{usd_ref_brask,usd_ref_ko} and is still a subject of both theoretical \cite{usd_ref_Kawakubo,usd_ref_bergou}, as well as experimental studies \cite{usd_ref_becerra,usd_ref_Agnew,usd_ref_brask}.

If the signal states are generated from an ensemble $\mathcal{E}=\left\{ p_{i},\rho_{i}\right\} _{i=1}^{n}$ and measured with a POVM $\M\in\PP(d,n+1)$,
success probability for USD is given by
\begin{equation}
p_{\mathrm{USD}}\left(\mathcal{E},\M\right)=\sum_{i=1}^{n}p_{i}\mathrm{tr}(\rho_{i}M_{i}),\,\label{eq:USDsuccPROB}
\end{equation}
where measurement effects have to satisfy the constraints $\mathrm{tr(\rho_{i}M_{j})=0}$,
for $i\neq j$, which result from the unambiguity condition. Moreover,
the effect $M_{n+1}$ corresponds to the inconclusive result. In this
Letter we focus on ensembles consisting of pure signal stares, i.e., 
$\rho_{i}=\kb{\psi_{i}}{\psi_{i}}$. In this case unambiguous discrimination
is possible if and only if vectors $\ket{\psi_{i}}$ are linearly
independent. Given an ensemble $\mathcal{E}$, we define $p_{\mathrm{USD}}^{\mathrm{POVM}}\left(\mathcal{E}\right)$
and $p_{\mathrm{USD}}^{\mathbb{SP}}\left(\mathcal{E}\right)$ as the
optimal success probabilities of unambiguously discriminating states
from $\mathcal{E}$ via generalized and projective simulable measurements
(acting on $n=d$ dimensional space spanned by vectors $\ket{\psi_{i}}$),
respectively. 
The following result limits the maximal advantage that POVMs can offer over PMs for USD.
\begin{lem}
\label{lem:boundSUCC}For all ensembles of linearly independent pure states $\mathcal{E}=\left\{ p_{i},\kb{\psi_i}{\psi_i}\right\} _{i=1}^{d}$ , we have $p_{\mathrm{USD}}^{\mathrm{POVM}}\left(\mathcal{E}\right)\leq d\cdot p_{\mathrm{USD}}^{\mathbb{SP}}\left(\mathcal{E}\right)$. 
\end{lem}
\begin{proof}
Let $\M^{\ast}=\left(M_{1}^{\ast},\ldots,M_{d+1}^{\ast}\right)$
be an optimal measurement for which $p_{\mathrm{USD}}\left(\mathcal{E},\M^{\ast}\right)=p_{\mathrm{USD}}^{\mathrm{POVM}}\left(\mathcal{E}\right)$.
We can use the protocol from the proof of Theorem \ref{thm:mainRES}
to construct the measurement $\M_{q}^{\ast}\in\mathbb{SP}(d,d+2)$.
By gluing outcomes $d+1$ and $d+2$ we get a projective simulable
measurement attaining success probability $\frac{1}{d}p_{\mathrm{USD}}^{\mathrm{POVM}}\left(\mathcal{E}\right)$.
Therefore, we have $\frac{1}{d}p_{\mathrm{USD}}^{\mathrm{POVM}}\left(\mathcal{E}\right)\leq p_{\mathrm{USD}}^{\mathbb{SP}}\left(\mathcal{E}\right)$.
\end{proof}
We now show that the above bound is essentially tight in the limit
of large $d$ by giving examples of ensembles $\mathcal{E}$ for which
$p_{\mathrm{USD}}^{\mathrm{POVM}}\left(\mathcal{E}\right)/p_{\mathrm{USD}}^{\mathbb{SP}}\left(\mathcal{E}\right)\approx d$.
We first state an auxiliary result that limits the power of projective
measurements for USD.
\begin{lem}
\label{lem:projSIMbound}Let $\mathcal{E}=\left\{ p_{i},\kb{\psi_{i}}{\psi_{i}}\right\} _{i=1}^{n}$
and let $\bk{\psi_{i}}{\psi_{j}}\ne0$ for $i\neq j$. Then we have
$p_{\mathrm{USD}}^{\mathbb{SP}}\left(\mathcal{E}\right)\le\max_{i}p_{i}$.
\end{lem}
\begin{proof}	
	We consider an ensemble  $\mathcal{E} = \{p_i, \ketbra{\psi_i}{\psi_i}\}_{i=1}^n$
	of quantum states in $\mathbb{C}^d$. We assume, that vectors $\ket{\psi_i}$ are 
	linearly independent and $\bk{\psi_i}{\psi_j} \neq 0$ for $i\neq j$. 
	Consider a projective measurement $\P=(P_1,\ldots,P_{n+1})$ on a subspace spanned by vectors 
	$\{\ket{\psi_i}\}_{i=1}^n$, i.e., $\H = \mathcal{L}in 
	(\{\ket{\psi_i}\}_{i=1}^n)$. We require $\P$ to satisfy unambiguity condition, i.e., for  $j \neq i$ we have $P_i \ket{\psi_j} = 0$. We will show that projective measurements satisfying above constraints are of 
	the form
	\begin{equation}
	\M^i = \left(
	\overbrace{0,0,\dots,0}^{i-1},
	\ketbra{\phi_i}{\phi_i},
	\overbrace{0,0,\dots ,0}^{n-i},
	\1-\ketbra{\phi_i}{\phi_i}
	\right) \label{eq:formPROJ} \ ,
	\end{equation}
	for some $j\leq n$ and a pure state $\ket{\phi_i}$ satisfying $\bk{\psi_j}{\phi_i}=0$ for $i\neq j$. 
	
	Assume that projector $P_i$ is non-zero -- then the unambiguity can be written in 
	terms of orthogonality with appropriate subspace
	\begin{equation}
	\mathrm{supp}(P_i) \perp \mathcal{L}in (\{\ket{\psi}\}_{j=1, j\neq i }^n).
	\end{equation} 
	In the above $\mathcal{L}in (\{\ket{\psi}\}_{j=1, j\neq i }^n)$ denotes a 
	linear subspace spanned by vectors $\{\ket{\psi}\}_{j=1, j\neq i }^n$.
	By the definition we have 
	\begin{equation}
	\mathrm{supp}(P_i) \subset \mathcal{H} = \mathcal{L}in (\{\ket{\psi}\}_{j=1}^n)\ ,
	\end{equation} 
	Thus we have obtained, that $\mathrm{supp}(P_i)$ is one--dimensional subspace  
	and therefore $P_i = \ketbra{\phi_i}{\phi_i}$.
	We will write vector $\ket{\psi_i}$ in terms of vector $\ket{\phi_i}$ and 
	$\ket{\psi_k}$, for some $k \neq i$, i.e., 
	\begin{equation}
	\ket{\psi_i} = \alpha \ket{\phi_i} + \beta \ket{\psi_k} + \gamma \ket{r},
	\end{equation}
	where $\ket{r} \in \mathcal{L}in ((\{\ket{\psi}\}_{j=1, j\neq i,k}^n)$ and 
	$\scalar{\psi_i}{r} =\scalar{\psi_k}{r} = 0$. Moreover 
	$|a|^2+|\beta|^2+|\gamma|^2=1$ and $\beta \neq 0$.
	Next we will show that projector $P_k$ must be zero. The USD property gives us 
	\begin{equation}
	0 = P_k \ket{\psi_i} = \alpha P_k \ket{\phi_i} + \beta P_k \ket{\psi_k} + 
	\gamma P_k \ket{r}.
	\end{equation}
	From our assumption the terms $P_k \ket{\phi_i}$ and $P_k \ket{r}$ are equal to 
	$0$. Thus we conclude that $P_k \ket{\phi_k}$ must be $0$. Therefore 
	$P_2 = 0$.
	
	Having shown Eq.\eqref{eq:formPROJ}, we see that all projective-simulable measurements that satisfy unambiguity condition are of the form $\N=\sum_{i=1}^n q_i \M^i $ , where $\lbrace q_i \rbrace_{i=1}^n$ is a probablility distribution. Therefore the success probability for such measurements can be bounded as follows
	\begin{equation}
	p_\mathrm{USD}(\mathcal{E},\N) \leq \max_i p_i \bra{\psi_i} P_i \ket{\psi_i} \leq \max_i p_i.
	\end{equation}
\end{proof}
\begin{exa}
\label{ex:Symm}Consider a uniform ensemble $\mathcal{E}_{\mathrm{\mathrm{sym}}}=\left\{ 1/d,\kb{\phi_{i}}{\phi_{i}}\right\} _{i=1}^{d}$
of so-called symmetric states in $\mathbb{C}^{d}$, i.e., states of the
form $\ket{\phi_{i}}=(1/\sqrt{d})\sum_{k=0}^{d-1}c_{k}\omega^{ik}\ket k$, where $\omega=\exp\left(\frac{2\pi i}{d}\right)$. In \cite{CHEFLES1998}
it was shown that $p_{\mathrm{USD}}^{\mathrm{POVM}}\left(\mathcal{E}_{\mathrm{sym}}\right)=d\min_{k}|c_{k}|^{2}$.
Since for any $\epsilon\in(0,1)$ we can set $\min_{k}|c_{k}|^{2}=(1-\epsilon)/d$.
We get that 
$p_{\mathrm{USD}}^{\mathrm{POVM}}\left(\mathcal{E}_{\mathrm{sym}}\right)=1-\epsilon$.
On the other hand, since for $\epsilon\in(0,1)$ we have $\bk{\phi_{i}}{\phi_{j}}\neq0$
and, by the virtue of Lemma \ref{lem:projSIMbound}, we get $p_{\mathrm{USD}}^{\mathrm{\mathbb{SP}}}\left(\mathcal{E}_{\mathrm{sym}}\right)\leq1/d$, 
and therefore (by combining with Lemma \ref{lem:boundSUCC}) we obtain
\begin{equation}
d(1-\epsilon)\leq p_{\mathrm{USD}}^{\mathrm{POVM}}\left(\mathcal{E}_{\mathrm{sym}}\right)/p_{\mathrm{USD}}^{\mathbb{SP}}\left(\mathcal{E}_{\mathrm{sym}}\right)\leq d\ .\label{eq:symmGAP}
\end{equation}
\end{exa}

The following example shows that the  inequality $p_{\mathrm{succ}}^{\mathrm{POVM}}\left(\mathcal{E}\right)\leq d\cdot p_{\mathrm{succ}}^{\mathbb{SP}}\left(\mathcal{E}\right)$ is saturated also for randomly chosen ensembles of quantum states.
\begin{exa}
\label{ex:Random}Consider a uniform ensemble $\mathcal{E}_{\mathrm{\mathrm{ran}}}=\left\{ 1/n,\kb{\varphi_{i}}{\varphi_{i}}\right\} _{i=1}^{d}$
of $d\leq D$ independently-chosen Haar-random pure states in $\mathbb{C}^{D}$
\footnote{By Haar-random pure states in $\mathbb{C}^{D}$ we understand random
states $\kb{\varphi}{\varphi}=U\kb{\varphi_{0}}{\varphi_{0}}U^{\dagger}$, where
$\kb{\varphi_{0}}{\varphi_{0}}$ is a fixed pure state from $\mathbb{C}^{D}$ and $U$
is $D\times D$ Haar-distributed unitary matrix.}. 
We are interested in values of the ratio of $p_{\mathrm{USD}}^{\mathrm{\mathrm{POVM}}}\left(\mathcal{E}_{\mathrm{\mathrm{ran}}}\right)/ p_{\mathrm{USD}}^{\mathbb{SP}}\left(\mathcal{E}_{\mathrm{\mathrm{ran}}}\right)$
that appear typically, i.e., with high probability over the choice
of states $\kb{\varphi_{i}}{\varphi_{i}}$. In Part A of SM \cite{supp} we show that in the limit $d,D\rightarrow \infty$
, while $d/D\rightarrow\gamma\in(0,1)$, with high probability we have
\begin{equation}
d\left(1-\gamma\right)^{2}\leq p_{\mathrm{USD}}^{\mathrm{POVM}}\left(\mathcal{E}_{\mathrm{sym}}\right)/p_{\mathrm{USD}}^{\mathbb{SP}}\left(\mathcal{E}_{\mathrm{sym}}\right)\leq d\ . \label{eq:typEXAMPL}
\end{equation}
Hence, for generic ensembles $\mathcal{E}_{\mathrm{\mathrm{ran}}}$
 the inequality from Lemma \ref{lem:boundSUCC} is asymptotically
saturated in the limit $d/D\rightarrow\gamma$ (up to the possible
correction $\left(1-\gamma\right)^{2}$). 
\end{exa}
The above considerations give a fairly complete
understanding of relative power of projective and generalized measurements
for USD in large dimensions. To our best knowledge, the only other
quantum task for which this kind of analysis was carried out is quantum
filtering \cite{Bergou2003}. The problem of  USD of random states have not been studied previously. So far the research efforts focused on minimal error discrimination \cite{Montanaro2007,Puchala2016} or on distinguishing between states that were altered by application of the infinitisimal unitary transformation \cite{OszmaniecMetro2016}.

\section{Illustration on IBM Quantum Processor}\label{sec:IBM}
IBM Q Experience is an online platform that allows to remotely perform experiments on IBM's quantum processors \cite{ibm_q_experience,qiskit,cross_qasm}.  
The devices themselves consist of superconducting transmon qubits \cite{koch_transmons}, which are manipulated via coupling to the external microwave field which, in principle, offers a full control over the qubits in a given processor.
In particular, this interaction allows to implement arbitrary one-qubit unitary and a two-qubit CNOT gate (via cross-resonance effect \cite{sheldon_cross_resonance}). 
By combining those gates with projective measurements in computational basis allowed on IBM's quantum devices, one is able to construct arbitrary two-qubit quantum circuit \cite{vatan_circuits}. 
We have used access to the 5-qubit quantum device \textit{ibqmx4} to implement three different POVMs on one qubit via our scheme and via well-known method of Naimark's construction \cite{Peres2006}. 
Three implemented POVMs were: 4-outcome tetrahedral \cite{Renes2003}, 3-outcome trine \cite{Jozsa2003}, and a randomly generated 4-outcome measurement. 
In what follows we describe in detail how such implementation proceeded.

Naimark's construction requires extension of a system of interest by an ancilla system. 
When we get an extended space, we construct an unitary on this space, which (followed by a measurement) implements desired POVM on our system. 
In the case of IBM's qubits, extension of space simply required usage of two, instead of a single qubit.
In order to implement a POVM on single qubit, we have implemented Naimark's unitary on two-qubit space and measured the system, obtaining the outcomes corresponding to those of a POVM.
We note that, while in principle to implement $n$-outcome measurement one needs extension to only $n$-dimensional space, in the case of qubits the actual Hilbert space dimension is restricted to be a power of $2$.
Therefore, to implement $3$-outcome measurement via Naimark's extension, we needed to construct $4$-dimensional unitary which is a direct sum of $3$-dimensional Naimark's unitary and the number $1$. 

The implementation by our scheme proceeded as follows.
As described in previous sections, for a measurement of the form
$\M=\left(\alpha_{1}\kb{\psi_{i}}{\psi_{i}},\ldots,\alpha_{n}\kb{\psi_{n}}{\psi_{n}}\right)$, the new scheme requires implementation of each of the projective measurement $\P^{j}=\left(\ketbra{\psi_{j}}{\psi_{j}},\mathds{1}-\ketbra{\psi_{j}}{\psi_{j}}\right)$ with probability equal to $\frac{\alpha_{j}}{2}$ (in the case of qubits) in each experimental run (sampling). In principle, it is possible to classically randomize choice of the projective measurement in each experimental run, accordingly to such probability distribution (as illustrated in detail in Example~\ref{ex::tetrahedral_implementation}). 
Unfortunately, on IBM's quantum device it would be practically infeasible to implement projective measurements once at the time, due to the fact that there is a time-limiting queue of jobs requested by users, and such randomization would require thousands of job requests. 
In order to overcome this obstacle, we have decided to simulate randomization by gaining statistics for every $\P^{j}$ from number of experimental runs proportional to $\alpha_{j}$. 
Since maximal number of experiments performed in one commissioned job on IBM Q Experience is $8192$, we have set number of experiments implementing projective measurement corresponding to eigenvalue $\alpha_{j}$ as $N_{j}\coloneqq 8192\  \frac{\alpha_{j}}{\max_{j}\alpha_{j}}$ ,
where $j\in \left\{1,2,\dots,n \right\}$. 
In other words, we have set the maximum probability to correspond to maximum number of experimental runs possible on IBM Q Experience. 
Then, we commissioned experiments for projective measurements with number of runs equal to or lower than the maximal one, in accordance to the values of $\alpha_{j}$'s. 
We note that such method of randomization is equivalent to sampling, provided we assume stability of the quantum device in time. 
Final step was to normalize all experimental counts to relative frequencies, simply by dividing all statistics by the number of runs for all projective measurements. 
It's worth adding that in the case of tetrahedral and trine POVMs, which are symmetric, all $\alpha_{j}$'s were the same, hence probability distributions were uniform for them. 

To compare the quality of both implementations we have performed quantum measurement tomography (QMT) \cite{Lundeen_QDT}, which is a procedure of tomographic reconstruction of POVM based on its implementation on the informationally-complete set of quantum states. 
Chosen set of quantum states consisted of both eigenstates of Pauli's $\sigma_{z}$ and two eigenstates of $\sigma_{x}$ and $\sigma_{y}$ corresponding to positive eigenvalues. 
We have reconstructed POVMs via linear inversion. 
We note here that in the case of our scheme, the statistics taken for the QMT were additionally normalized only to non-rejected (correct) outcomes.
Details of QMT, together with remarks regarding the impact of noise, are provided in the Appendix~\ref{app:QMT}.
The explicit form of both POVMs to-be-implemented and the reconstructed ones are given in the Appendix~\ref{app:povms_explicit}.

As a figure of merit we used the operational distance \cite{zbyszek_tv_distance} between POVMs $\M$ and $\N$, which may be calculated as
\begin{equation}\label{eq::dop}
D_{op}\rbracket{\M,\N}=\max_{x\in X}\ ||\sum_{i\in x} \rbracket{M_i-N_i}||\ ,
\end{equation}
where $||\cdot||$ denotes the operator norm and the maximization is over all combinations of indices enumerating effects (i.e., all outcomes).
Naturally, in this case $\M$ corresponds to ideal POVM to-be-implemented, and $\N$ to its tomographic reconstruction.
The operational distance given in \eq{eq::dop} has a nice operational interpretation of $D_{op}(\M,\N)=2p_{\mathrm{dist}}(\M,\N) -1$, where $p_{\mathrm{dist}}(\M,\N)$ is the optimal probability of distinguishing between measurements $\M$ and $\N$ \emph{without using entanglement} \cite{Sedlak2014,zbyszek_tv_distance}.
The results of experiments are given in Table \ref{table_QMT}.

\begin{table}[h]
\begin{tabular}{|l|c|c|}
\hline
\multicolumn{1}{|c|}{POVM} & Naimark construction & Our scheme \\ \hline
Tetrahedral                & 0.117                  & 0.023              \\ \hline
Trine                      & 0.141                  & 0.022              \\ \hline
Random 4-effect            & 0.168                  & 0.031              \\ \hline
\end{tabular}
\caption{ Operational distances $D_{op}$ between POVMs to-be-implemented and those obtained via QMT for both methods of implementation.}
\label{table_QMT}
\end{table}

It is clear that our scheme performs better than the Naimark's construction. We would like to stress that this is the case despite rejecting half of the data (this results from postselection used in our method).  A likely explanation of these results is the much greater amount of noise occurring in the implementation of two-qubit unitaries required for Naimark's construction compared to local unitaries needed for our scheme (see pictorial representation of exemplary quantum circuits in Appendix~\ref{app:circuits}). 

\section{Discussion}\label{sec:Concl}

We heve presented a new method for implementing generalized measurements
on finite-dimensional Hilbert spaces that uses only classical resources
(randomization and post-processing), projective measurements and postselection.
Importantly, the scheme does not require to implement projective measurements
on extended Hilbert space. This simplification comes at the
expense or probabilistic nature of the method - in a given experimental
run it succeeds with probability $1/d$, where $d$ is the dimension
of the system. We have also used this result to find a (saturable)
upper bound on the relative power of POVMs and projective measurements
for the problem of unambiguous state discrimination.

We believe that our scheme will be useful in experimental implementations
of generalized measurements. Complicated global unitaries, that are
necessary in the Naimark construction, often introduce additional
errors. We have observed this kind of behavior in the experiments
carried out on IBM Q Experience platform. Of course, if a given
experimental setup allows to reliably implement unitary operations
on the extended system, then our method will not be beneficial. However, we expect that for the inherently noisy near-term quantum devices  \cite{NISC2017},  our scheme might prove advantageous over the standard method that requires extension of the Hilbert space. 

At the end we would like to state a number of possible directions
of further research. First, it would be interesting to explore if the techniques presented here can be used to show that Bell nonlocality with respect to projective measurements is equivalent to Bell nonlocality with respect to POVMs \cite{Barrett2002,Brunner2014}, despite the fact that postselection performed in Bell scenario can be used to violate Bell inequalities by local models \cite{Gisin99}. Second, it is interesting to relate the probability, with which a given generalized measurement $\M$ can be simulated using projective measurements and postselection, with the amount of white noise that is necessary to simulate with projective measurements \cite{Oszmaniec2017}. 
Also, it is intriguing to connect probability of success with \emph{entanglement cost} of generalized measurements \cite{Jozsa2003}. 
Another interesting  question will be to ask how postselection of measurements can be meaningfully used in infinite-dimensional setting. 
Last but not least, from the foundational perspective it is natural to explore the role of postselection for measurements in the General Probabilistic Theories  that go beyond quantum mechanics (a recent work \cite{Filippov2018} studied the role classical randomization and post-processing exactly in this context).

\begin{acknowledgments}
We thank Antonio Ac{\'i}n, Anubhav Chaturvedi,  Marco T\'ulio Quintino and Andreas Winter for
interesting and stimulating discussions. We are especially grateful
to Robert Griffiths for suggesting to use postselection together with
projective measurements to simulate POVMs.  M.O and F.B.M acknowledge the support
of Homing programme of the Foundation for Polish Science co-financed
by the European Union under the European Regional Development Fund. Z.P acknowledges the support of NCN grant 2016/22/E/ST6/00062.
We acknowledge use of the IBM Quantum Experience for this work. The
views expressed are those of the authors and do not reflect the official
policy or position of IBM or the IBM Quantum Experience team.
\end{acknowledgments}

\bibliography{refmeas1}

\onecolumngrid
\appendix
\section{Details regarding Example 3}\label{app::details_ex_3}

In the main text we did not prove the technical claims given in Example 3. Here we present a justification of \eq{eq:symmGAP}. 

By the seminal results of Eldar \cite{Eldar2003} we know that for any uniform ensemble of quantum states $\ket{\psi_{i}}$ there exists so-called ``equal probability measurement''  $\M_{\mathrm{eq}}$ (this measurement $\M_{\mathrm{eq}}$ can be regarded as the analogue of the ``pretty-good measurement'' \cite{Barnett09} \footnote{For the details of the construction of $\M_{\mathrm{eq}}$ see Section 4 of \cite{Eldar2003}. } used, e.g, in the context of minimal-error state discrimination . The measurement  $\M_{\mathrm{eq}}$ attains the success probability \begin{equation}
p_{\mathrm{USD}}\left(\mathcal{E},\mathcal{\M_{\mathrm{eq}}}\right)=\lambda_{\mathrm{min}}\left(C\right)\ ,
\end{equation} 
where $\lambda_{\mathrm{min}}(C)$ is the minimal eigenvalue of the $d\times d$ correlation matrix $C_{ij}=\bk{\psi_{i}}{\psi_{j}}$. Therefore, for the problem at hand we get the lower bound $p_{\mathrm{USD}}^{\mathrm{POVM}}\left(\mathcal{E}_{\mathrm{\mathrm{ran}}}\right)\geq\lambda_{\mathrm{min}}(C_{\mathrm{ran}})$ with $C_{\mathrm{ran}}$ being the correlation matrix of Haar-random unit vectors $\ket{\varphi_{i}}$ form $\C^D$. The minimal eigenvalue of $C_{\mathrm{ran}}$ has been studied in the mathematical literature \cite{RanomCorr2004} and for typical ensembles $\mathcal{E}_{\mathrm{\mathrm{ran}}}$ in the limit: $d,D\rightarrow\infty$, while $d/D\rightarrow\gamma\in(0,1)$, we have $\lambda_{\mathrm{min}}(C_{\mathrm{rand}})\approx\left(1-\gamma\right)^{2}$. In order to bound the success probability $p_{\mathrm{USD}}^{\mathrm{\mathbb{SP}}}\left(\mathcal{E}_{\mathrm{ran}}\right)$ we note that generic Haar random vectors $\left\{ \ket{\varphi_{i}}\right\} _{i=1}^{d}$
are linearly independent but not orthogonal and therefore, by the virtue of Lemma 2, we have $p_{\mathrm{USD}}^{\mathrm{\mathbb{SP}}}\left(\mathcal{E}_{\mathrm{ran}}\right)\leq1/d$. On the other hand we have $p_{\mathrm{USD}}^{\mathbb{SP}}\left(\mathcal{E}_{\mathrm{ran}}\right)\leq(1/d)p_{\mathrm{USD}}^{\mathrm{POVM}}\left(\mathcal{E}_{\mathrm{\mathrm{ran}}}\right)$. 
Combining this with $\left(1-\gamma\right)^{2}\leq p_{\mathrm{USD}}^{\mathrm{POVM}}\left(\mathcal{E}_{\mathrm{\mathrm{ran}}}\right)\leq 1$, we finally obtain 
\begin{equation}
d\left(1-\gamma\right)^{2}\leq p_{\mathrm{USD}}^{\mathrm{POVM}}\left(\mathcal{E}_{\mathrm{sym}}\right)/p_{\mathrm{USD}}^{\mathbb{SP}}\left(\mathcal{E}_{\mathrm{sym}}\right)\leq d\label{eq:typEXAMPL}\ .
\end{equation}

\section{Quantum Measurement Tomography}\label{app:QMT}
In order to reconstruct measurement done on the quantum system, one needs to perform a detector tomography. 
In this section we describe a simple method we chose to do so for implemented measurements.
Every two-dimensional effect $M_{i}$ can be written in the form of re-scaled Bloch vector
\begin{align}
M_{i}=\frac{\alpha_{i}}{2} \left(\mathds{1}+\vec{n}_{i}\vec{\sigma}\right) \ ,
\end{align}
where $\alpha_{i}\in\left(0,1\right]$ and $|\vec{n}|\leq1$. To obtain value of $\alpha_{i}$ and three components of real vector $\vec{n}_{i}$, we used a Born's rule 
$ p_{i}=Tr\left(M_{i}\rho\right)=Tr\left(\frac{\alpha_{i}}{2} \left(\mathds{1}+\vec{n}_{i}\vec{\sigma}\right)\rho\right)
$ and a freedom in choosing initial state $\rho$.

We performed four experiments for four different quantum states, which we chose to be eigenstates of Pauli matrices. 
Two of them were both eigenstates of $\sigma_{z}$, while two other were eigenstates of $\sigma_{x}$ and $\sigma_{y}$ corresponding to positive eigenvalues. 
Elementary calculations show that to obtain $\alpha_{i}$ one can add up statistics obtained for both eigenstates of $\sigma_{z}$. 
Having this value calculated, in order to obtain components of Bloch vector, one has to simply transform Born's rule equation and use statistics obtained for eigenstates of all Pauli matrices corresponding to positive eigenvalues.

By repeating the above procedure for all effects of POVM $\M$, one can reconstruct the whole measurement.
In general, this method may result in reconstructing unphysical, non-positive operators.
At the beginning of our work with IBM Q devices, we have been surprised that such thing has never occur, nor in Naimark's case, nor in our scheme method. We note that in principle, in order to avoid reconstruction of unphysical operators, one needs to implement optimization algorithms, such as in \cite{lundeen_QDT}. The systematic positivity of obtained operators probably results from the nature of noise in IBM Q devices.

The analysis of errors occuring during the measurements in IBM devices lies outside the scope of this work and will be the subject of the future work \cite{Maciejewski2018}. 
However, already at this point we describe a few subtle issues we encountered while dealing with readout errors.  
We have noticed, that a systematic error occurs in IBM Q devices, namely there exists a constant bias towards obtaining the '0' result for any kind of qubit circuit. 
Natural explanation of such an error might be a decoherence occuring during a readout. For our scheme of implementation, we firstly identified a '0' as a non-postselected result and a '1' as a postselected one. 
Due to the bias error, it resulted in all cases in postselection on average on \textit{more} than $1/2$ results. 
To fight this bias, we have doubled a number of implemented circuits, and for half of them we simply applied an $x$ gate and relabeled the outcomes. 
The data obtained from circuits with additional $x$ gate resulted in postselection on average on \textit{less} than $1/2$ of the results. 
Finally, averaged data from standard circuits and circuits with $x$ gate resulted in postselection on average on around $1/2$ data and always in reconstruction of positive operators. 

The analogous method was used in the Naimark implementation, where there are 4 possible $x$ gates configurations - i) no $x$ gates, ii) $x$ gate only on first qubit, iii) $x$ gate only on second qubit, iv) $x$ gates on both qubits. This procedure also led to reconstruction in which all effects were positive operators.

In practice, to compare Naimark's construction with our scheme, we needed equal numbers of experiments to calculate probabilities $p_{i}$. 
As described in the main text, maximum number of experiments in a single job request on IBM Q experience is $8192$, while chosen method of randomization required $N>8192$ runs, which exact value was dependant on eigenvalues of effects.
To compare measurements implemented by both methods, for Naimark's construction we have performed around (up to divisibility of $N$ by $3$ or $4$) N experiments.  After gaining statistics we calculated the operational distance given in \eq{eq::dop}
It's quite interesting to note that for 3-outcome measurements, when theoretically there should be only 3 possible outcomes, in experimental realisation via Naimark's construction there were always some additional clicks on the 4th outcome. 
In QMT this resulted in the appearance of the 4th "residual" effect, which was, naturally, taken into account for computation of $D_{op}(\M,\M_{exp})$.

\section{Generalised measurements for QMT experiments}\label{app:povms_explicit}

In this section, we provide explicit matrix forms of all to-be-implemented POVMs and the ones reconstructed via method described above. The reconstructed matrix elements are given with numerical precision of 3 digits.

\textbf{Tetrahedral POVM }\cite{Renes04Crypt}

\noindent Tetrahedral measurement $\M_{\mathrm{tetra}}$ consist of effects with Bloch vectors pointing to vertices of tetrahedron inscribed in the Bloch sphere, 
\begin{align}
M_{1}= \begin{bmatrix}%
\frac{1}{2}&0\\[6pt]
0&0
\end{bmatrix}\ ,\   M_{2}= \begin{bmatrix}%
\frac{1}{6}&\frac{1}{3\sqrt{2}}\\[6pt]
\frac{1}{3\sqrt{2}}&\frac{1}{3}
\end{bmatrix}\ ,\   
M_{3}= \begin{bmatrix}%
\frac{1}{6}&\frac{1}{3\sqrt{2}}e^{-i\frac{2\pi}{3}}\\[6pt]
\frac{1}{3\sqrt{2}}e^{+i\frac{2\pi}{3}}&\frac{1}{3}
\end{bmatrix}\ ,\   M_{4}= \begin{bmatrix}%
\frac{1}{6}&\frac{1}{3\sqrt{2}}e^{+i\frac{2\pi}{3}}\\[6pt]
\frac{1}{3\sqrt{2}}e^{-i\frac{2\pi}{3}}&\frac{1}{3}
\end{bmatrix}\ .
\end{align}

\noindent Quantum Measurement Tomography for implementation using our scheme resulted in reconstruction of following effects:
\begin{align}
M_{1}= \begin{bmatrix}%
0.489 &-0.007+0.007i\\[6pt]
-0.007-0.007i&0.016
\end{bmatrix}\ ,\   M_{2}= \begin{bmatrix}%
0.167&0.226-0.003i\\[6pt]
0.226+0.003i&0.327
\end{bmatrix}\ ,\   
\end{align}

\begin{align} 
M_{3}= \begin{bmatrix}%
0.169&-0.107-0.195i\\[6pt]
-0.107+0.195i&0.330
\end{bmatrix}\ ,\   M_{4}= \begin{bmatrix}%
0.175&-0.112+0.191i\\[6pt]
-0.112-0.191i&0.327
\end{bmatrix}\ .
\end{align}

\noindent Quantum Measurement Tomography for implementation using Naimark's construction resulted in reconstruction of following effects:
\begin{align}
M_{1}= \begin{bmatrix}%
0.462 & -0.025-0.013i\\[6pt]
-0.025+0.013i&0.052
\end{bmatrix}\ ,\   M_{2}= \begin{bmatrix}%
0.169&0.167-0.017i\\[6pt]
0.167+0.017i&0.282
\end{bmatrix}\ ,\   
\end{align}

\begin{align} 
M_{3}= \begin{bmatrix}%
0.187&-0.079-0.162i\\[6pt]
-0.079+0.162i&0.294
\end{bmatrix}\ ,\   M_{4}= \begin{bmatrix}%
0.182&-0.062+0.192i\\[6pt]
-0.062-0.192i&0.371
\end{bmatrix}\ .
\end{align}

\textbf{Trine POVM} \cite{trine_povm}

\noindent Trine measurement $\M_{\mathrm{trine}}$ consist of effects with Bloch vectors pointing to vertices of equilateral triangle inscribed in the Bloch sphere,
\begin{align}
M_{1}= \begin{bmatrix}%
\frac{2}{3}&0\\[6pt]
0&0%
\end{bmatrix}\ ,\   M_{2}= \begin{bmatrix}%
\frac{1}{6}&\frac{1}{2\sqrt{3}}\\[6pt]
\frac{1}{2\sqrt{3}}&\frac{1}{2}
\end{bmatrix}\ ,\   M_{3}= \begin{bmatrix}%
\frac{1}{6}&-\frac{1}{2\sqrt{3}}\\[6pt]
-\frac{1}{2\sqrt{3}}&\frac{1}{2}
\end{bmatrix}\ .
\end{align}
\textit{}
\noindent Quantum Measurement Tomography for implementation using our scheme resulted in reconstruction of following effects:
\begin{align}
M_{1}= \begin{bmatrix}%
0.645&-0.004+0.004i\\[6pt]
-0.004-0.004i&0.021
\end{bmatrix}\ ,\   M_{2}= \begin{bmatrix}%
0.178&0.272-0.002i\\[6pt]
0.272+0.002i&0.489
\end{bmatrix}\ ,\ 
\end{align}

\begin{align}  
M_{3}= \begin{bmatrix}%
0.177&-0.268-0.001i\\[6pt]
-0.268+0.001i&0.490
\end{bmatrix}\ .
\end{align}

\noindent Quantum Measurement Tomography for implementation using Naimark's construction resulted in reconstruction of following effects:
\begin{align}
M_{1}= \begin{bmatrix}%
0.599&0.003-0.021i\\[6pt]
0.003+0.021i&0.072
\end{bmatrix}\ ,\   M_{2}= \begin{bmatrix}%
0.192&0.210+0.004i\\[6pt]
0.210-0.004i&0.403
\end{bmatrix}\ ,\   
\end{align}

\begin{align}  
M_{3}= \begin{bmatrix}%
0.170&-0.224+0.019i\\[6pt]
-0.224-0.019i&0.460
\end{bmatrix}\ ,\   M_{4}= \begin{bmatrix}%
0.038&0.011-0.003i\\[6pt]
0.011+0.003i&0.065
\end{bmatrix}\ .
\end{align}
Note additional, fourth effect mentioned in the previous section.

\textbf{Random 4-outcome POVM}

\noindent Last of implemented POVMs was constructed randomly. To generate a random POVM, we constructed Haar-random 4-dimensional unitary and took first two elements of each column to define vectors that were used as effects of 4-outcome POVM.
Matrix elements are given with numerical precision of 3 digits,
\begin{align}
M_{1}= \begin{bmatrix}%
0.288&0.061-0.049i\\[6pt]
0.061+0.049i&0.021
\end{bmatrix}\ ,\   M_{2}= \begin{bmatrix}%
0.063&0.070-0.109i\\[6pt]
0.070+0.109i&0.264
\end{bmatrix},
\end{align}

\begin{align}
M_{3}= \begin{bmatrix}%
0.470&0.17-0.002i\\[6pt]
0.17+0.002i&0.062
\end{bmatrix}\ ,\   M_{4}= \begin{bmatrix}%
0.179&-0.301+0.160i\\[6pt]
-0.301-0.160i&0.653
\end{bmatrix}\ .
\end{align}

\noindent Quantum Measurement Tomography for implementation using our scheme resulted in reconstruction of following effects:
\begin{align}
M_{1}= \begin{bmatrix}%
0.281&0.054-0.046i\\[6pt]
0.054+0.046i&0.029
\end{bmatrix}\ ,\   M_{2}= \begin{bmatrix}%
0.068&0.068-0.101i\\[6pt]
0.068+0.101i&0.261
\end{bmatrix}\ ,\   
\end{align}

\begin{align}
M_{3}= \begin{bmatrix}%
0.455&0.160-0.002i\\[6pt]
0.160+0.002i&0.078
\end{bmatrix}\ ,\   M_{4}= \begin{bmatrix}%
0.196&-0.282+0.149i\\[6pt]
-0.282-0.149i&0.632
\end{bmatrix}\ .
\end{align}

\noindent Quantum Measurement Tomography for implementation using Naimark's construction resulted in reconstruction of following effects:
\begin{align}
M_{1}= \begin{bmatrix}%
0.313&0.060-0.044i\\[6pt]
0.060+0.044i&0.064
\end{bmatrix}\ ,\   M_{2}= \begin{bmatrix}%
0.089&0.038-0.079i\\[6pt]
0.038+0.079i&0.303
\end{bmatrix}\ ,\   
\end{align}

\begin{align}  
M_{3}= \begin{bmatrix}%
0.411&0.129+0.009i\\[6pt]
0.129-0.009i&0.106
\end{bmatrix}\ ,\   M_{4}= \begin{bmatrix}%
0.187&-0.227+0.114i\\[6pt]
-0.227-0.114i&0.528
\end{bmatrix}\ .
\end{align}

\section{General quantum circuits used in implementation}\label{app:circuits}
In this section we provide a pictorial presentation of quantum circuits needed for implementation of general qubit POVMs using Naimark's construction method and our scheme. In the figures presented below, $\mathrm{SU}(2)$ denotes arbitrary local unitary on qubit and $H$ is a Hadamard gate. Quantum register is denoted by $qr$ with subscript being a label of a qubit, whereas classical register is merged into one line denoted by $cr$.

\begin{figure}[h]
	\begin{centering}
		\includegraphics[width=1\textwidth]{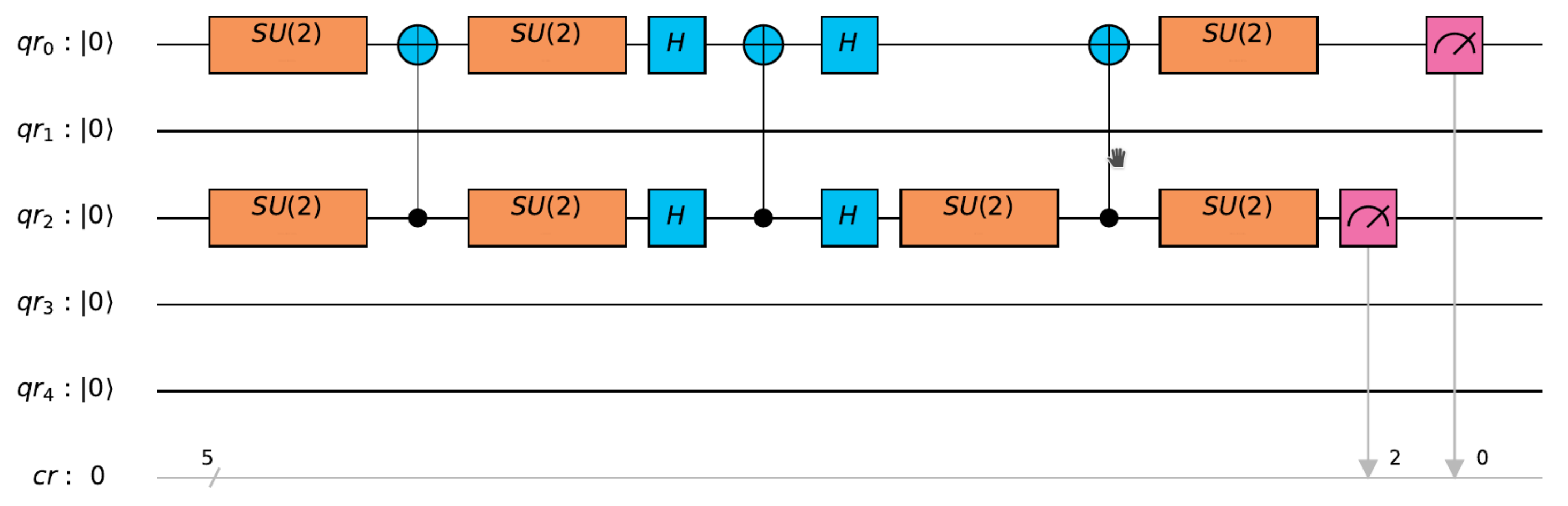}\caption{\label{fig:Circuit_naimark} Quantum circuit needed to implement POVM using Naimark's construction method. Additional Hadamard gates at the center are needed due to one-way connectivity of CNOT gates in IBM hardware.}
		\par\end{centering}
\end{figure}

\begin{figure}[h]
	\begin{centering}
		\includegraphics[width=0.3\textwidth]{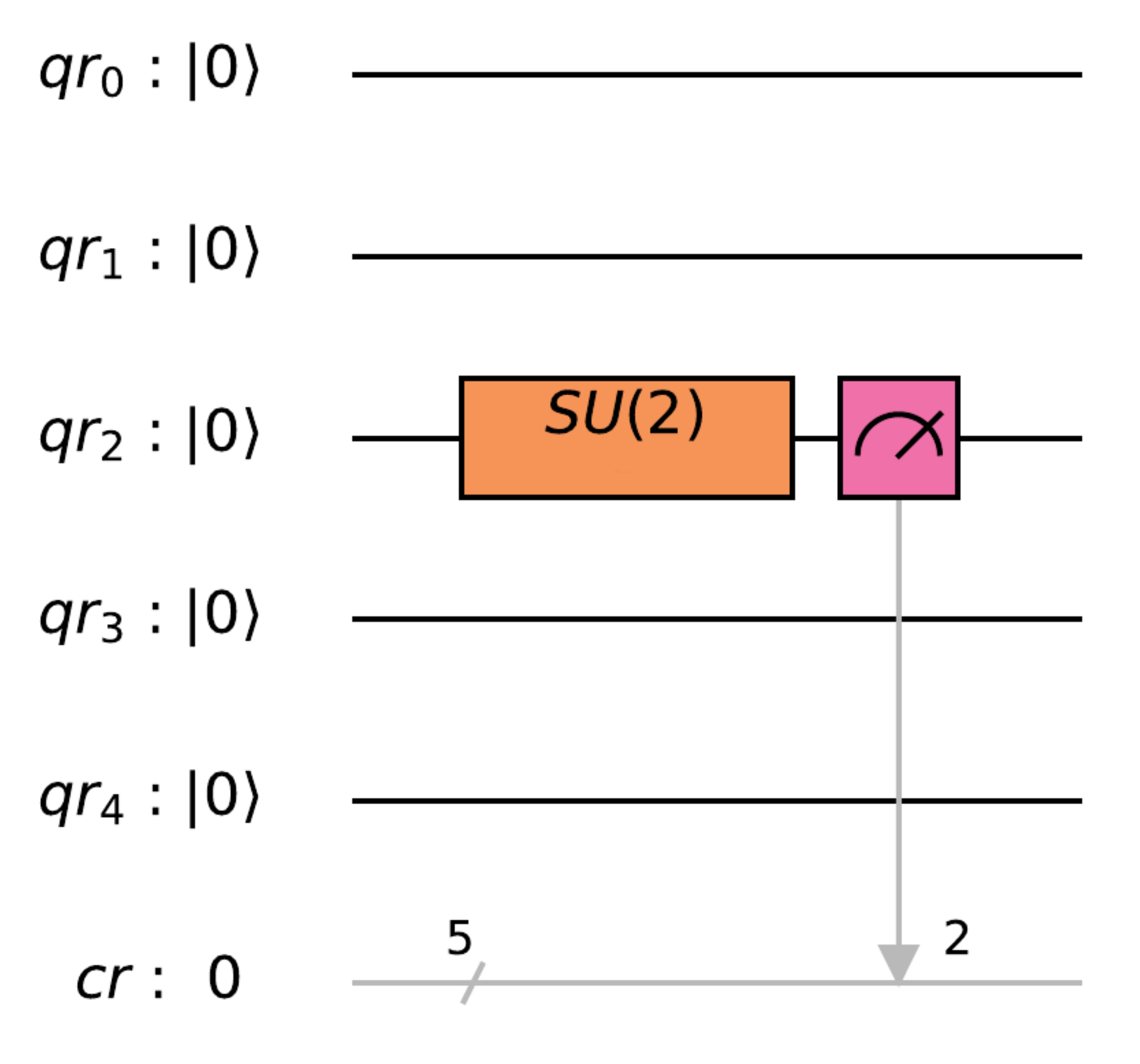}\caption{\label{fig:Circuit_randpost} Quantum circuit needed to implement POVM using our scheme.}
		\par\end{centering}
\end{figure}

In principle, in the qubit case Naimark's construction requires implementation of a \textit{single} of circuit of the general form shown in Figure \ref{fig:Circuit_naimark}, whereas our scheme requires implementation of \textit{multiple} (equal to the number of outcomes of simulated POVM) circuits of the form shown in Figure \ref{fig:Circuit_randpost}.

\end{document}